\documentclass[11pt, a4paper]{article}

\usepackage[margin=1in]{geometry}

\usepackage{amsmath,amsfonts,amsthm,amssymb}
\usepackage[retainorgcmds]{IEEEtrantools}
\usepackage{xcolor}
\usepackage{soul}
\usepackage[utf8]{inputenc}
\usepackage[small]{caption}
\usepackage{tikz}
\usetikzlibrary{cd}
\usepackage{enumitem}
\usepackage{booktabs}
\usepackage{xspace}
\usepackage{footnote}
\usepackage{fourier}

\theoremstyle{plain}
\newtheorem{theorem}{Theorem}[section]

\newtheorem{lemma}[theorem]{Lemma}
\newtheorem{proposition}[theorem]{Proposition}

\newtheorem{definition}[theorem]{Definition}
\newtheorem{corollary}[theorem]{Corollary}

\theoremstyle{definition}
\newtheorem{example}{Example}

\newcommand{\bmu}{\ensuremath{\boldsymbol{\mu}}\xspace}
\newcommand{\mms}{\ensuremath{\textsc{mms}}\xspace}
\newcommand{\pmms}{\ensuremath{\textsc{pmms}}\xspace}
\newcommand{\efo}{\ensuremath{\textsc{ef}\oldstylenums{1}}\xspace}
\newcommand{\efx}{\ensuremath{\textsc{efx}}\xspace}
\newcommand{\ef}{\ensuremath{\textsc{ef}}\xspace}
\newcommand{\mysetminusD}{\hbox{\tikz{\draw[line width=0.6pt,line cap=round] (3pt,0) -- (0,6pt);}}}
\newcommand{\mysetminusT}{\mysetminusD}
\newcommand{\mysetminusS}{\hbox{\tikz{\draw[line width=0.45pt,line cap=round] (2pt,0) -- (0,4pt);}}}
\newcommand{\mysetminusSS}{\hbox{\tikz{\draw[line width=0.4pt,line cap=round] (1.5pt,0) -- (0,3pt);}}}
\newcommand{\mysetminus}{\mathbin{\mathchoice{\mysetminusD}{\mysetminusT}{\mysetminusS}{\mysetminusSS}}}

\DeclareMathOperator*{\argmin}{arg\,min}
\DeclareMathOperator*{\argmax}{arg\,max}

\allowdisplaybreaks

\title{Comparing Approximate Relaxations of Envy-Freeness}
\author{
	Georgios Amanatidis$^1$\hspace{10pt} 
	Georgios Birmpas$^2$\hspace{10pt} 
	Evangelos Markakis$^2$  \vspace{5pt}
	\\ 
	$^1$ \normalsize{Centrum Wiskunde \& Informatica (CWI), the Netherlands} \\
	$^{2}$ \normalsize{Department of Informatics, Athens University of Economics and Business, Greece} \vspace{2pt} \\
	\normalsize{\texttt{amanatid@cwi.nl\hspace{10pt}
	gebirbas@aueb.gr\hspace{10pt}
	markakis@aueb.gr}}
}

\begin{document}

\maketitle

\begin{abstract}
In fair division problems with indivisible goods it is well known that one cannot have any guarantees for the classic fairness notions of envy-freeness and proportionality. As a result, several relaxations have been introduced, most of which in quite recent works. We focus on four such notions, namely envy-freeness up to one good (\efo), envy-freeness up to any good (\efx), maximin share fairness (\mms), and pairwise maximin share fairness (\pmms). Since obtaining these relaxations also turns out to be problematic in several scenarios, approximate versions of them have been considered. 
In this work, we investigate further the connections  between the four notions mentioned above and their approximate versions. 
We establish several tight, or almost tight, results concerning the approximation quality that any of these notions guarantees for the others, providing an almost complete picture of this landscape. Some of our findings reveal interesting and surprising consequences regarding the power of these notions, e.g., \pmms and \efx provide the same worst-case guarantee for \mms, despite \pmms being a strictly stronger notion than \efx. We believe such implications provide further insight on the quality of approximately fair solutions. 
\end{abstract}

\section{Introduction}
\label{sec:intro}

In this work, we revisit fairness notions for allocating indivisible goods. 
The objective in fair division is to allocate a set of resources to a set of agents in a way that leaves everyone satisfied, according to their own preferences.
Over the years, the field has grown along various directions, with a substantial literature by now and with several applications.
We refer the reader to the surveys \cite{BCM16-survey,Procaccia16-survey,Markakis17-survey} for more recent results, and to the classic textbooks \cite{BT96,RW98,Moulin03} for an overview of the area.

To model such allocation problems, one needs to specify the preferences of the agents, and the fairness criterion under consideration. For the preferences, we stick to the usual assumption of associating  each agent with an {\em additive} valuation function on the set of resources. As for fairness criteria, one of the classic  desirable notions that have been proposed is \textit{envy-freeness}, meaning that no agent has a higher value for the bundle of another agent than for her own \cite{Foley67,Varian74}. Unfortunately, for problems with indivisible goods, this turns out to be a very strong requirement. Existence of envy-free allocations cannot be guaranteed, and the relevant algorithmic and approximability questions are also computationally hard \cite{LMMS04}.

Given these concerns, recent works have considered relaxations of envy-freeness that seem more appropriate for settings with indivisible items, see, e.g., \cite{Budish11,CKMP016}. Our work focuses on four such notions, namely {envy-freeness up to one good} (\efo), {envy-freeness up to any good} (\efx), {maximin share fairness} (\mms), and {pairwise maximin share fairness} (\pmms). 
All these capture different ways of allowing envy in an allocation, albeit in a restricted way. For instance, \efo 
requires that no agent envies another agent's bundle after removing from it her most 
valued item.

These relaxations are still no panacea, especially since existence results have either been elusive or simply negative. So far, we only know that \efo allocations always exist, whereas this is not true for \mms allocations. As a result, this has naturally led to approximate versions of these notions, accompanied by some positive algorithmic results (see related work). What has not been well charted yet, however, is the relation between these four notions and their approximate counterparts, especially concerning the approximation quality that any of these notions guarantees for the others. For example, does an \mms or an approximate \mms allocation (for which we do know efficient algorithms) provide an approximation guarantee in terms of the \efx or the \pmms criteria? As it turns out (Prop.~\ref{prop:a_mms_to_pmms}, Cor. \ref{cor:a_mms_to_efx}), it does not. As another example, while we know that \pmms implies \efx, it is not clear if an approximate \pmms allocation is also approximately \efx. In fact (Prop.~\ref{prop:a_efx_to_pmms} and \ref{prop:a_pmms_to_efx}), it is the other way around! 
Such results can be conceptually  helpful, as they deepen our understanding of the similarities and differences between fairness criteria. Furthermore, this insight allows us to either translate an approximation algorithm for one notion into an approximation algorithm for another, or to establish that such an approach cannot yield approximability or existential results. Given the growing interest in these relaxations, we find it important to further explore these interconnections. 
\medskip

\noindent {\bf Contribution.}
We investigate how the four notions mentioned above and their approximate versions are related. For each one of them, we examine the approximation guarantee that it provides in terms of the other three notions. Our results provide an almost complete mapping of the landscape, and in most cases our approximation guarantees are tight. Some of our results provide interesting and surprising consequences. Indicatively, some highlights of our results include:
\begin{itemize}[label={--},itemsep=0.5ex]
	\item \pmms and \efx allocations both provide the same constant approximation with respect to the \mms criterion.
	\item Although \pmms implies \efx, an approximate \pmms allocation may be arbitrarily bad in terms of approximate \efx. On the contrary, an approximate \efx allocation does provide some guarantee with respect to \pmms.  
	\item While \efx is a much stronger concept than \efo, they both provide comparable constant approximations for the \pmms criterion and this degrades smoothly for approximate \efx and \efo allocations.
	\item Although exact \pmms and \mms are defined in a similar manner, the former implies  \efx, \efo, and a $4/7$-approximation in terms of \mms, while the latter provides no guarantee with respect to the other notions. 
	\item Our results also suggest a simple  efficient algorithm for computing a $1/2$-approximate \pmms allocation (the current best to our knowledge),\footnote{Independently, the recent work of Barman et al.~\cite{BBMN17}, which concerns a stronger concept, also implies an efficient algorithm for computing $1/2$-\pmms allocations.} and improvements in certain special cases.  
\end{itemize}
Some of the implications between the different notions are depicted in Figure \ref{fig:implications}. 
\medskip

\noindent {\bf Related Work.} Regarding the several relaxations of envy-freeness, the notion of \efo originates in the work of Lipton et al.~\cite{LMMS04}, where both existence and algorithmic results are provided. The concept of \mms was introduced by Budish \cite{Budish11}, building on concepts of Moulin  \cite{Moulin90}, and later defined as considered here by Bouveret and Lema\^{i}tre \cite{BL16}, who studied a hierarchy of exact fairness concepts. Kurokawa, Procaccia and Wang \cite{KurokawaPW18,KPW16} showed that \mms allocations do not always exist even for three agents. From the computational point of view, $2/3$-approximation algorithms for \mms allocations are known \cite{KurokawaPW18,AMNS17,BM17} and recently this approximation factor has been improved to $3/4$ by Ghodsi et al.~\cite{SGHSY18} who also study non-additive valuation functions. Some variants of the maximin share criterion have also been considered for groups of agents \cite{Suk18}.
As for the notions of \efx and \pmms, they were introduced in the work of Caragiannis et al.~\cite{CKMP016} which provided some initial results on how these are related to each other as well as to \mms and to \efo. In particular, they established that a \pmms allocation is also an \efx, \efo and $\frac{1}{2}$-\mms allocation. 
Barman and Krishnamurty \cite{BM17} showed that when the agents agree on the ordering of the items according to their value, then \efx allocations do exist and a  specific \efx allocation  can be produced that is also a $\frac{2}{3}$-\mms allocation. 
The existence and computation of exact and approximate \efx allocations under additive and general valuations was studied by Plaut and Roughgarden \cite{PR18}, establishing the currently best $1/2$-approximation.
Finally, in a  recent work, Barman et al.~\cite{BBMN17} introduce two new fairness notions that are strongly related to the notions studied here. In particular, they introduce \emph{groupwise maximin share fairness} (\textsc{gmms}), which is a strengthening of \pmms, and \emph{envy-freeness up to one less-preferred good} (\textsc{efl}), which is  stronger than \efo but weaker than \efx. They show that \textsc{efl} allocations and $\frac{1}{2}$-\textsc{gmms} allocations always exist and can be found in polynomial time. Studying how the approximate versions of  \textsc{gmms} and  \textsc{efl} fit into the landscape explored here is an interesting direction for future work.

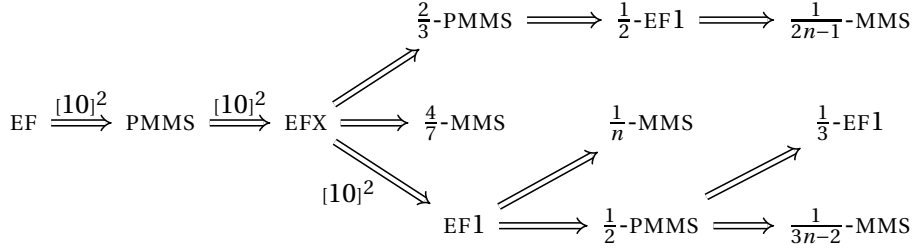
\begin{figure}[t]
	\begin{center}
			\begin{tikzcd}[column sep=normal, row sep=normal]
				& & & \frac{2}{3}\text{-\pmms} \arrow[r, Rightarrow] & \frac{1}{2}\text{-\efo} \arrow[r, Rightarrow] & \frac{1}{2n-1}\text{-\mms} \\
				\ef \arrow[r, Rightarrow, "\cite{CKMP016}\footnotemark{}"] & \pmms \arrow[r, Rightarrow,"\cite{CKMP016}\protect{\footnotemark[\value{footnote}]}"] &	\efx  \arrow[ru, Rightarrow] \arrow[rd, Rightarrow, swap, "\cite{CKMP016}\protect{\footnotemark[\value{footnote}]}"] \arrow[r, Rightarrow] & \frac{4}{7}\text{-\mms}  & \frac{1}{n}\text{-\mms} & \frac{1}{3}\text{-\efo}  \\
				& & & \efo \arrow[r, Rightarrow] \arrow[ru, Rightarrow] & \frac{1}{2}\text{-\pmms} \arrow[r, Rightarrow] \arrow[ru, Rightarrow] & \frac{1}{3n-2}\text{-\mms}
		\end{tikzcd}
		\caption{An indicative chart of implications with envy-freeness as a starting point. All implications shown are tight or almost tight.} \label{fig:implications}
	\end{center}
\end{figure}

\footnotetext{The implications of Caragiannis et al.~\cite{CKMP016} apply only to the exact versions of the corresponding notions.}

\section{Preliminaries}
\label{sec:prelim}

We assume we have a set of $n$ agents, $N = \{1, 2, \ldots, n\}$ and a set $M$ of $m$ indivisible goods. 
Following the usual assumptions in the majority of the fair division literature, each agent is associated with a monotone, {\it additive} valuation function. Hence, for every $S\subseteq M$, 
$v_i(S) = \sum_{g\in S} v_i(\{g\})$. For simplicity, we will use $v_{i}(g)$ instead of $v_i(\{g\})$ for $g\in M$. Monotonicity here implies that $v_{i}(g)\geq 0$ for every $i\in N, g\in M$.

We are interested in solutions that allocate the whole set of goods $M$ to the agents. An allocation of $M$ to the $n$ agents is therefore a partition,
$\mathcal{A} = (A_1,\ldots,A_n)$, where $A_i\cap A_j = \emptyset$ and $\bigcup_i A_i = M$.
By $\Pi_n(M)$ we denote  the set of all partitions of a set $M$ into $n$ bundles.

\subsection{Fairness Concepts}
\label{subsec:concepts}
Our work focuses on relaxations of the classic notion of envy-freeness, initially suggested by Gamow and Stern \cite{GS58}, and more formally by Foley \cite{Foley67} and Varian \cite{Varian74}.

\begin{definition}
	\label{def:EF}
	An allocation $\mathcal{A} = (A_1,\ldots,A_n)$ is envy-free (\ef), if for every $i, j\in N$, $v_i(A_i) \geq v_i(A_j)$.
\end{definition}

It is well known that envy-freeness is a very strict requirement in the presence of indivisible goods; e.g., consider any instance where all agents have large value for one specific good and negligible value for everything else. This fact gives rise to considering relaxations of envy-freeness, with the hope of obtaining more positive results.

We begin with two such relaxations, and their approximate versions, where an agent may envy some other agent, but only by an amount dependent on the value of a single item in the other agent's bundle. Formally:  
\begin{definition}
	\label{def:EF1-EFX}
	Let $\alpha \in [0,1]$.
	An allocation $\mathcal{A} = (A_1,\ldots,A_n)$ is an
	\begin{itemize}[label={--}]
		\item $\alpha$-\efo ($\alpha$-envy-free up to one good) allocation, if for every pair of agents $i, j\in N$, with $A_j\neq\emptyset$, there exists an item $g\in A_j$, such that
		$v_i(A_i) \geq \alpha\cdot v_i(A_j\setminus \{g\})$.
		\item $\alpha$-\efx ($\alpha$-envy-free up to any good) allocation, if $v_i(A_i) \geq \alpha\cdot v_i(A_j\setminus \{g\})$ holds for every pair $i, j\in N$, with $A_j\neq\emptyset$, and for every  $g\in A_j$ with $v_{i}(g) >0$.\footnote{The requirement that $v_{i}(g) >0$ in the definition of $\alpha$-\efx has been dropped by Plaut and Roughgarden \cite{PR18} but several of their results hold under the  assumption of all values being strictly positive.} 
	\end{itemize}
\end{definition}

Note that for $\alpha=1$ in the above definitions, we obtain precisely the notions of envy-freeness up to one good (\efo) \cite{Budish11} and envy-freeness up to any good (\efx) \cite{CKMP016}. The difference between these two notions is simply the quantifier regarding the item that eliminates envy when removed from an agent's bundle. Clearly, \ef implies \efx, which in turn implies \efo. 

We now move on to a different notion, also proposed by Budish \cite{Budish11}.  
Motivated by the question of what we can guarantee in the worst case to the agents, the rationale is to think of a generalization of the well-known cut-and-choose protocol to multiple agents: Suppose that agent $i$ is asked to partition the goods into $n$ bundles and then the rest of the agents choose a bundle before $i$. In the worst case, agent $i$ will be left with her least valuable bundle. Hence, a risk-averse agent would choose a partition that maximizes the minimum value of a 
bundle in the partition. This gives rise to the following definition.

\begin{definition}
	\label{def:mmshare}
	Given $n$ agents, and a subset  $S\subseteq M$ of items, the $n$-maximin share of agent $i$ with respect to $S$ is:
	\[ \bmu_i(n, S) = \displaystyle\max_{\mathcal{A}\in\Pi_n(S)} \min_{A_j\in \mathcal{A}} v_i(A_j)\,.\]
\end{definition}
From the definition, it directly follows that $n\cdot \bmu_i(n, S)\le v_i(S)$.
When $S=M$, this quantity is called the \textit{maximin share} of agent $i$. 
We say that $\mathcal{T} \in \Pi_n(M)$ is an $n$-maximin share defining partition for agent $i$, if $\min_{T_j\in \mathcal{T}} v_i(T_j) = \bmu_i(n, M)$.
When it is clear from context what $n$ and $M$ are, we will simply write $\bmu_i$ instead of $\bmu_i(n, M)$. 
The solution concept we are interested in, asks for a partition that gives each agent her (approximate) maximin share.
\begin{definition}
	\label{def:MMS}
	Let $\alpha \in [0,1]$.
	An allocation 
	$\mathcal{A} = (A_1,\ldots,A_n) $ is called an $\alpha$-\mms ($\alpha$-maximin share) allocation if $v_i(A_i)\geq \alpha\cdot \bmu_i\,$, for every $i\in N$.
\end{definition} 

The last notion we define is related but not directly comparable to \mms. The idea is that instead of imagining an agent $i$ partitioning the items into $n$ bundles, we think of $i$ as partitioning the combined bundle of herself and another agent into two bundles and receiving the one she values less. 

\begin{definition}
	\label{def:PMMS}
		Let $\alpha \in [0,1]$.
	An allocation $\mathcal{A} = (A_1,\ldots,A_n) $ is called an $\alpha$-\pmms ($\alpha$-pairwise maximin share) allocation if for every pair of agents $i, j\in N$, 
	\[v_i(A_i)\geq \alpha\cdot \! \max_{(B_1, B_2)} \! \min \{v_i(B_1), v_i(B_2) \}\,,\] 
	where $(B_1, B_2)\in \Pi_2(A_i\cup A_j)$.
\end{definition} 

In Definitions \ref{def:MMS} and \ref{def:PMMS}, when $\alpha=1$, we refer to the corresponding allocations as \mms and \pmms allocations respectively.
It has been already observed that the notion of \pmms is stronger than \efx  \cite{CKMP016}. 

\begin{example}
	\label{ex:illustration}
	To illustrate the relevant definitions, let us consider an instance with 3 agents and 5 items: 
	\begin{center}
		\begin{tabular}{@{} l | *5l @{}}    
			& $a$ & $b$ & $c$ & $d$ & $e$\ \  \\\toprule
			\ Agent 1 & $3$ & $1$ & $1$ & $1$ & $4$\ \  \\ 
			\ Agent 2 & $4$ & $3$ & $3$ & $1$ & $4$\ \  \\ 
			\ Agent 3 & $3$ & $2$ & $1$ & $3$ & $4$\ \  \\\bottomrule
			
		\end{tabular}
	\end{center}\smallskip
	
	If $M = \{a, b, c, d, e\}$ is the set of items, one can see that $\bmu_1(3, M) = 3$, $\bmu_2(3, M) = 4$, $\bmu_3(3, M) = 4$. For example, looking at agent $1$, there exists a partition so that the worst bundle is worth a value of $3$, and there is no partition where the worst bundle is better.
	Similarly, agent $2$ can produce a partition where the worst bundle has a value of $4$ for her, and there is no other partition that can guarantee a better worst-case value.
	
	Let us examine the allocation $\mathcal{A} = (\{e\}, \{b, c\}, \{a, d\})$. Clearly, this is an \ef allocation, and hence it is also an \efx, \efo, \mms, and \pmms allocation. Consider now the allocation $\mathcal{B} = (\{a\}, \{b, e\}, \{c, d\})$. This is no longer \ef and it is also neither \efx nor \pmms. However, it is an \efo as well as an \mms allocation. Clearly, each agent $i$ receives a value of at least $\bmu_i(3, M)$. To see why it is \efo, note that agents 1 and 3 envy agent 2 but removing item $e$ from the bundle of agent $2$ eliminates the envy from either agent. We can also see that $\mathcal{B}$ is a $\frac{3}{4}$-\efx allocation. To verify this, we can look at agent 1, and compare the value of her bundle to the bundle of agent $2$ when we remove either item $b$ or $e$. Finally, it is not hard to check that $\mathcal{B}$ is also a $\frac{3}{4}$-\pmms allocation.        
\end{example}

\section{EFX and EF1 Allocations}
\label{sec:EFX-EF1}

We begin our exposition with the more direct relaxations of envy-freeness, \efo and \efx. Within this section, we always start with either an $\alpha$-\efo or $\alpha$-\efx allocation, for some $\alpha\in (0, 1]$, and explore the implications and fairness guarantees that can be derived. We also pay particular attention to the case of exact \efx or \efo allocations, i.e., for $\alpha=1$. 

We already know that \efx is stronger than \efo. Our first warm-up proposition states that this also holds in an approximation preserving sense. We complement this by the fact that \efo allocations cannot provide any approximation to \efx. 
\begin{proposition}\label{prop:efx_to_efo}
	For $n\ge 2$, any $\alpha$-\efx allocation is also an $\alpha$-\efo allocation for any $\alpha \in (0,1]$. On the other hand, an \efo allocation is not necessarily an $\alpha$-\efx allocation for any $\alpha \in (0,1]$.
\end{proposition}

\begin{proof}
	Start with an $\alpha$-\efx allocation $\mathcal{A} = (A_1, \ldots, A_n)$. This means that for any $i,j \in N$ we have that $v_i(A_i) \geq \alpha\cdot v_j(A_j\mysetminus g)$, where $g \in \argmin_{h\in A_j, v_i(h)>0}v_i(h)$. However this also means that for any $i,j \in N$, there exists some $g \in A_j$,  such that $v_i(A_i) \geq \alpha\cdot v_j(A_j\mysetminus g)$. Thus $\mathcal{A}$ is also an $\alpha$-\efo allocation.
	
	Now regarding the second part of the statement, consider the following simple example: Suppose that we have $2$ agents and $3$ items, $g_1, g_2, g_3$. The agents have identical values over the items, specifically $v_i(g_1)=V,v_i(g_2)=1$ and $v_i(g_3)=\alpha$, for  $i \in \{1,2\}$, $\alpha \in (0,1]$ and $V\gg 1$. We now consider the allocation $\mathcal{A}=(\{g_1,g_2\}, \{g_3\})$. It is easy to see that this is an $\alpha$-\efo allocation but also an $\frac{1}{V}$-\efx allocation, and $\frac{1}{V}$ can be made arbitrarily small. 
	By adding any number of dummy agents and an equal number of dummy items, so that  $g_1, g_2, g_3$ have no value for the dummy agents and the $j$th dummy item has nonzero value only for the $j$th dummy agent,
	this instance can be easily generalized to any number of agents.
\end{proof}

    Notice that an equivalent way of stating the $\efo \not\Rightarrow \alpha\text{-}\efx$ part of Proposition \ref{prop:efx_to_efo} would be: \emph{an $\alpha$-\efo allocation is not necessarily a $\beta$-\efx allocation for any $\alpha,\beta \in (0,1]$}.

Before proving guarantees in terms of \mms and \pmms, we state a useful technical lemma which generalizes both (the $k=1$ case of) Lemma 1  of Bouveret and Lema\^{i}tre \cite{BL16} and Lemma 3.4 of Amanatidis et al.~\cite{AMNS17}.
   In its simplest form, the lemma states that removing any item, together with any agent, does not decrease the maximin share of the remaining agents. In general, it allows to remove appropriately chosen subsets of items, while reducing the number of agents by 1, so that  the maximin share of a specific agent does not decrease.

\begin{lemma}\label{lem:monotonicity}
	Suppose $\mathcal{T} \in \Pi_n(M)$ is an $n$-maximin share defining partition for agent $i$, i.e., $\bmu_i(n, M) = \min_{T_j\in \mathcal{T}} v_i(T_j)$. Then, for any set of goods $S$, such that there exists some $j$ with $S \subseteq T_j$, it holds that 
	\[\bmu_i(n-1, M\mysetminus S) \geq \bmu_i(n, M)\,.\]
	In particular, for any item $g$, $\bmu_i(n-1, M\mysetminus \{g\}) \geq \bmu_i(n, M)$.
\end{lemma}

\begin{proof}
	Let us look at agent $i$, and consider a partition of $M$ that attains her maximin share.
	Let $\mathcal{T} = (T_1, \ldots,T_n)$ and assume, without loss of generality, that $S \subseteq T_n$. Consider the remaining partition $(T_1,\ldots,T_{n-1})$ enhanced in an arbitrary way by the items of $T_n \mysetminus S$.
	This is an $(n-1)$-partition of $M\mysetminus S$, where the value of agent $i$ for any bundle is at least $\bmu_i(n, M)$. Thus, 
	we have $\bmu_i(n-1, M\mysetminus S) \geq \bmu_i(n, M)$.
\end{proof}

Next we study \efx allocations in terms of the \mms guarantees they can provide, starting with the case of a small number of agents.

\begin{proposition}\label{prop:efx_to_mms_n=2,3}
	For $n\in\{2, 3\}$, an \efx allocation is also a $\frac{2}{3}$-\mms allocation. 
	Moreover, this guarantee is tight.
\end{proposition}

\begin{proof}
	We prove the statement for $n=2$. The proof for the case of three agents is slightly more complex but of similar flavor. 
	
	Suppose that we have  an \efx allocation $\mathcal{A}=(A_1, A_2)$.
	We show that for agent 1 we have $v_1(A_1) \ge \frac{2}{3} \bmu_1(2, M)$. The case of agent 2 is symmetric. Since items of no value for agent 1 are completely irrelevant for her view of both \efx and \mms allocations, we can assume without loss of generality, that $v_1(g)>0$ for all $g\in M$. 
	
	In the sequel, we write $\bmu_1$ instead of $\bmu_1(n, M)$. If $A_2=\emptyset$, then clearly $v_1(A_1) = v_1(M) \ge \bmu_1$, so we may assume $|A_2|\ge 1$. If $|A_2| = \{g\}$, then $v_1(A_1) = v_1(M\mysetminus\{g\}) = \bmu_1(1, M\mysetminus\{g\})\ge \bmu_1(2, M) =\bmu_1$, where the  inequality follows from Lemma \ref{lem:monotonicity}. 
	The remaining case is when $|A_2|\geq 2$. Suppose, towards a contradiction, that $v_1(A_1)< \frac{2}{3}\bmu_1$. Since $\mathcal{A}$ is an \efx allocation, we have that 
	\begin{equation}\label{eq_efx}
	v_1(A_2) \leq v_1(A_1) + v_{1}(g)\,, 
	\end{equation}
	where $g \in\argmin_{h\in A_2} v_1(h)$.
	Since $A_2$ contains at least two items, we have 
	\begin{equation}\label{eq_card}
	v_{1}(g) \leq \frac{1}{2}v_1(A_2)\,. 
	\end{equation}
	Combining \eqref{eq_efx} and \eqref{eq_card} we get $v_{1}(g) \leq v_1(A_1)$. Again by \eqref{eq_efx}, this implies that  $v_1(A_2)<\frac{4}{3}\bmu_1$. But then, $v_1(M)=v_1(A_1)+v_1(A_2)<\frac{6}{3}\bmu_1=2\bmu_1$, a contradiction, since by definition, we know that $v_i(M) \geq n\cdot \bmu_i$ for every $i\in N$.

	Regarding tightness, consider the following simple example. Suppose that we have $2$ agents and $4$ items, $a, b,c, d$. The agents have identical values over the items, specifically $v_i(a)=v_i(b)=2$ and $v_i(c)=v_i(d)=1$, for  $i \in \{1,2\}$. We now look at the allocation $\mathcal{A}=(\{a,b\}, \{c,d\})$. It is easy to see that $\bmu_1 = \bmu_2 = 3$ and that $\mathcal{A}$ an \efx allocation. However $v_2(\{c,d\}) = 2 = \frac{2}{3}\bmu_2$. By adding an arbitrary number of copies of agent $2$ and an equal number of items of value $3$, this instance can be generalized to any number of agents.
\end{proof}

Beyond three agents, the picture gets somewhat more complicated. The next result is one of the main  highlights of our work. When moving from the case of $n\leq 3$, to a higher number of agents, the approximation guarantee achieved, in terms of \mms, degrades from $2/3$ to a quantity between $4/7=0.5714$ and $0.5914$. Surprisingly, the same happens to the \mms guarantee of a \pmms allocation, as we show in the next section.

It is interesting to note  that recently Barman and Krishnamurty \cite{BM17} obtained a simple $2/3$-appro\-xi\-mation algorithm for \mms by showing that when agents agree on the ordering of the values of the items, then there exists a particular \efx allocation that is also a $2/3$-\mms allocation. As indicated by the proof of the following proposition, this is not true for all \efx allocations, even when the agents are identical.

\begin{proposition}\label{prop:efx_to_mms}
	For $n\ge 4$, any \efx allocation is also a $\frac{4}{7}$-\mms allocation. On the other hand, an \efx allocation is not necessarily an $\alpha$-\mms allocation for $\alpha > \frac{8}{13}$ and, as $n$ grows large, for $\alpha > 0.5914$.
\end{proposition}

\begin{proof}
	Let $\mathcal{A} = (A_1, \ldots, A_n)$ be an \efx allocation. Suppose, towards a 
	contradiction, that $\mathcal{A}$ is not a $\frac{4}{7}$-\mms allocation, i.e., 
	there exists some  $j$ so that $v_j(A_j)< \frac{4}{7} \bmu_j(n, M)$.
	Without loss of generality, we assume that agent $1$ is such an agent, and write 
	$\bmu_1$ instead of $\bmu_1(n, M)$. 
	Note that we may remove any agent, \mbox{other} than agent 1, that receives exactly one good, 
	and still end up with an \efx allocation that is not a $\frac{4}{7}$-\mms allocation. 
	Indeed, if $|A_i|=1$ for some $i>1$, then $(A_1, \ldots, A_{i-1}, A_{i+1}, \ldots, A_n)$ 
	is an \efx allocation of $M\mysetminus A_i$ to $N\mysetminus\{i\}$ and, 
	by Lemma \ref{lem:monotonicity}, $\bmu_1(n-1, M\mysetminus A_i) \geq \bmu_1(n, M)$. 
	Thus, $v_1(A_1)< \frac{4}{7} \bmu_1(n-1, M\mysetminus A_i)$.
	Therefore, again without loss of generality, we may assume that $|A_i|>1$ for all $i>1$ 
	in the initial allocation $\mathcal{A}$. 
	
	Given $\mathcal{A}$, we say that a bundle $A_j$ is \textit{bad} if $j>1$, $|A_j|=2$, and 
	$\min_{g\in A_j}v_1(g)> \frac{1}{2}v_1(A_1)$. An item is bad if it belongs to a bad bundle, while a  bundle is \textit{good} if it is not bad. Let $B$ be the set of all bad items.
	
	It is not hard to show that if $A_i$ is good, then $v_1(A_1) \ge \frac{2}{3}v_1(A_i)$. 
	When $i = 1$ this is straightforward; otherwise we consider two cases. 
	First assume that $|A_i| = 2$. Then, by definition, we have $\min_{g\in A_i}v_1(g)\le \frac{1}{2}v_1(A_1)$ and, since $\mathcal{A}$ is \efx, we have $\max_{g\in A_i}v_1(g)\le v_1(A_1)$. Thus, $v_1(A_i)\le \frac{3}{2} v_1(A_1)$.
	Next, assume that $|A_i| \ge 3$. Then, $\min_{g\in A_i}v_1(g) \le \frac{1}{3} v_1(A_i)$ and, since $\mathcal{A}$ is \efx, we have $v_1(A_i) - \min_{g\in A_i}v_1(g) \le v_1(A_1)$. Thus, we get $v_1(A_1) \ge v_1(A_i) - \frac{1}{3} v_1(A_i) = \frac{2}{3} v_1(A_i)$.
	
	Now we are going to show that $v_1(A_1) \ge \frac{4}{7} \bmu_1(n', M')$ for a reduced instance that we get by possibly removing some of the items of $B$, i.e., bad items. We do so in a way that ensures that $\bmu_1(n', M')\ge \bmu_1$, thus contradicting  the choices of $\mathcal{A}$ and $A_1$. 
	We consider an $n$-maximin share defining partition $\mathcal{T}$ for agent $1$, i.e., $\min_{T_i\in \mathcal{T}} v_1(T_i) = \bmu_1$. 
	If there is a bundle of $\mathcal{T}$ containing two items of $B$, $g_1$, $g_2$, then we remove those two items and reduce the 
	number of agents by one. By Lemma \ref{lem:monotonicity}, we have that $\bmu_1(n-1, M\mysetminus \{g_1, g_2\}) \ge \bmu_1$.
	We repeat as many times as necessary to get a reduced instance with $n'\le n$ agents and a set of 
	items $M'\subseteq M$ for which there is an $n'$-maximin share defining partition $\mathcal{T'}$ for agent $1$, such that no bundle contains more than $1$ item of $B$.  By repeatedly using Lemma \ref{lem:monotonicity}, we have $\bmu_1(n', M') \geq \bmu_1$. 
	
	Let $x$ be the number of items of $B$ in the reduced instance (recall that these items are defined with respect to $\mathcal{A}$, hence they belong to a bundle of size $2$ in $\mathcal{A}$ and have value at least $\frac{1}{2}v_1(A_1)$ for agent $1$). Clearly, $x$ cannot be greater than $n'$, or some bundle of $\mathcal{T'}$ would contain at least $2$ bad items. Further, if $|B| = y$, i.e., the number of bad items in the original instance, then we know that the number of good bundles in $\mathcal{A}$ was $n-\frac{y}{2}$, and that the number of agents was reduced $\frac{y-x}{2}$ times, i.e., $n'= n-\frac{y-x}{2}$. Hence, we may express the number of good bundles in the original instance in terms of $n'$ and $x$ only, as $n'-\frac{x}{2}$.
	
	Recall that $\bmu_1 \le \bmu_1(n', M')$ and, by the definition of maxi\-min share, $\bmu_1(n', M') \le \frac{1}{n'}v_1(M')$. Thus 
	\begin{equation}\label{eq:mms<prop}
	\bmu_1 \le \frac{1}{n'}v_1(M')\,.
	\end{equation}
	In order to upper bound $v_1(M')$, notice that $M'$ contains all the items of all the good bundles of $\mathcal{A}$ plus $x$ bad items. As discussed above, the value (with respect to agent 1) of each good bundle is upper bounded by $\frac{3}{2} v_1(A_1)$. On the other hand, $\mathcal{A}$ being \efx implies that $\max_{g\in A_i}v_1(g)\le v_1(A_1)$ for any bad bundle $A_i$. That is, any bad item's value is upper bounded by $v_1(A_1)$. So, we have
	\begin{IEEEeqnarray}{rCl}
		v_1(M') & \le & x\cdot v_1(A_1) + \left( n'-\frac{x}{2} -1\right) \frac{3}{2} v_1(A_1) + v_1(A_1) \nonumber\\
		& = & \left( \frac{3n'}{2}+\frac{x}{4} -\frac{1}{2}\right)  v_1(A_1) \nonumber\\
		& \le & \left( \frac{3n'}{2}+\frac{n'}{4} -\frac{1}{2}\right)  v_1(A_1) \nonumber\\
		& = & \frac{7n'-2}{4}  v_1(A_1)\,. \label{eq:4/7_prop}
	\end{IEEEeqnarray} 
	
	We combine inequalities \eqref{eq:mms<prop} and \eqref{eq:4/7_prop} to get $v_1(A_1) \ge \frac{4n'}{7n'-2} \bmu_1\ge \frac{4}{7} \bmu_1$, which contradicts the choices of $\mathcal{A}$ and $A_1$.	
	
	This establishes the positive part of the theorem. To see that \efx does not imply anything stronger than $\frac{8}{13}$-\mms, notice that for $n'=n=4$ the above analysis can be tight. On one hand, $\frac{4n'}{7n'-2} = \frac{8}{13}$ in this case, while on the other hand the following instance indicates that this is the best one can guarantee for 4 agents. Suppose that we have 12 items with the following values for agent 1: 
	\[ 
	v_1(g_i)= \left\{
	\begin{array}{ll}
	1/8 & 1\le i \le 4 \\
	1/2 & 5 \le i\le 8\\
	1 & 9 \leq i\leq 12 \\
	\end{array} 
	\right. 
	\]
	It is not hard to see that $\bmu_1 = 1 + \frac{1}{2} + \frac{1}{8} =\frac{13}{8}$ in this instance. Now consider the allocation $\mathcal{A} = (\{g_1, \ldots, g_5\}, \allowbreak \{g_6, g_7, g_8\}, \allowbreak \{g_9, g_{10}\}, \allowbreak \{g_{11}, \allowbreak  g_{12}\})$. Assuming that agents $2$, $3$ and $4$ are identical to  agent $1$, 
	it is not hard to check that this is an \efx allocation. Yet, $v_1(A_1) = 1 = \frac{8}{13} \bmu_1$. By adding an arbitrary number of copies of agent $4$ and her bundle, this instance can be generalized to any number of agents.
	
	The stronger bound for $\alpha$ as $n$ grows large follows by Propositions \ref{prop:pmms_to_efx} and \ref{prop:pmms_to_mms}.
\end{proof}

By following a similar analysis as in the proof of Proposition \ref{prop:efx_to_mms_n=2,3}, it can be shown that for any number of agents $\alpha$-\efx implies 
at least $\frac{\alpha}{2}$-\mms. The actual guarantee, moreover, cannot be much better than this.

\begin{proposition}\label{prop:a_efx_to_mms}
	For $n\geq 2$ and $\alpha \in (0,1)$, any $\alpha$-\efx allocation is also an $\frac{\alpha n}{\alpha+2n-2}$-\mms allocation but not necessarily a $\beta$-\mms allocation, for any $\beta> \frac{2\alpha}{2+\alpha}$. For $n\ge 4$, the upper bound is improved to $\max\{\frac{\alpha}{1+\alpha}, \frac{8\alpha}{11+2\alpha}\}$. 
\end{proposition}

\begin{proof}
	Let $\mathcal{A} = (A_1, \ldots, A_n)$ be an $\alpha$-\efx allocation and focus on agent 1. By Lemma \ref{lem:monotonicity}, we may assume, without loss of generality, that for every $i \geq 2$ we have that $|A_i| \geq 2$. 
	Thus, for all  $i \geq 2$, there exists some $g_i \in \argmin_{h \in A_i} v_1(h)$ so that $v_1(A_1) \geq \alpha (v_1(A_i)-v_1(g_i))$ and $v_1(g_i) \leq v_1(A_i \mysetminus g_i)$. Then 
	\begin{IEEEeqnarray*}{rCl}
		v_1(A_i) & \leq & \frac{1}{\alpha}v_1(A_1)+v_1(g_i) \leq \frac{1}{\alpha}v_1(A_1)+v_1(A_i \mysetminus g_i) \\
		& \leq & \frac{1}{\alpha}v_1(A_1) +\frac{1}{\alpha}v_1(A_1)= \frac{2}{\alpha}v_1(A_1)\,.
	\end{IEEEeqnarray*} 
	So, we have that 
	\begin{IEEEeqnarray*}{rCl}
		n\bmu_1 & \leq & v(M)= v(A_1)+ \sum_{i\geq2}v(A_i ) \\
		& \leq & v(A_1)+ \sum_{i\geq2}\frac{2}{\alpha}v_1(A_1)= \frac{\alpha+ 2(n-1)}{\alpha} v_1(A_1)\,,
	\end{IEEEeqnarray*} 
	and finally we can conclude 
	\[v_1(A_1) \geq \frac{\alpha n}{\alpha+2n-2}\bmu_1\,.\]
	
	Regarding the upper bound, fix any $\alpha \in (0,1)$ and suppose that we have $2$ agents and $4$ items: $a, b,c, d$. The agents have identical valuation functions, in particular $v_i(a)=v_i(b)=1$ and $v_i(c)=v_i(d)=\frac{\alpha}{2}$, for  $i \in \{1,2\}$. We now look at the allocation $\mathcal{A}=(\{a,b\}, \{c,d\})$. It is easy to see that this an $\alpha$-\efx allocation which, however, is only an $\frac{2\alpha}{2+\alpha}$-\mms allocation.
	
	For the case of $n \geq 4$, consider the following example. We define an instance with $n=4$ identical agents and $k+7$ items. Let $\alpha \in (0,1]$. For every agent $i$ we have 
	\[ 
	v_i(g_j)= \left\{
	\begin{array}{ll}
	\frac{\alpha}{k} & 1 \leq j\leq k\\
	\frac{1}{2} & k+1 \leq j\leq k+3\\
	1 & k+4 \leq j\leq k+7\\
	\end{array} 
	\right. 
	\]
	Let $\mathcal{A} = (A_1, \ldots, A_4)=(\{g_1,g_2,\ldots,g_k\}, \{g_{k+1}, g_{k+2}, g_{k+3},\}, \{g_{k+4},g_{k+5}\},\{g_{k+6},g_{k+7}\})$. It is easy to see that this is an $\alpha$-\efx allocation which is only an $\frac{\alpha}{1+\alpha}$-\mms allocation when $\alpha<\frac{1}{2}$, and  is only an $\frac{8\alpha}{11+2\alpha}$-\mms allocation when $\alpha\geq\frac{1}{2}$ and $k$ is large.
	
	It is not hard to see that both instances above can be generalized to any number of agents.
\end{proof}

In contrast to \efx, $\alpha$-\efo allocations provide a much weaker approximation of maximin shares, namely a ratio of $O\big(\frac{1}{n}\big)$ for constant $\alpha$. The following proposition generalizes a result of Caragiannis et al.~\cite{CKMP016}.\footnote{The bound by Caragiannis et al.~\cite{CKMP016} is for $\alpha=1$ and  for allocations that are both \efo and Pareto optimal. It follows by their proof, however, that it holds even when  Pareto optimality is dropped.}

\begin{proposition}\label{prop:efo_to_mms}
	For $n\geq 2$,  any $\alpha$-\efo allocation is also an $\frac{\alpha}{n-1+\alpha}$-\mms allocation for any $\alpha \in (0,1]$, and this is tight.
\end{proposition} 

\begin{proof}
	Let $\mathcal{A} = (A_1, \ldots, A_n)$ be an $\alpha$-\efo allocation and fix some agent $i$.  The case where at least one of the bundles is empty is straightforward, so we may assume that the bundles are all non-empty. For all agents $j 
	\in N\mysetminus\{i\}$,  there exists an item $g_{j}\in A_j$ such that  $v_i(A_i)\geq \alpha(v_i(A_j)- v_i(g_{j}))$. By summing up over $N\mysetminus\{i\}$ and adding $\alpha v_i(A_i)$ to both sides, we have   
	\[(n-1)v_i(A_i) +\alpha v_i(A_i) \geq \alpha\bigg( v_i(A_i) + \sum_{j \neq i}v_i(A_j)- \sum_{j \neq i }v_i(g_{j}) \bigg) =  \alpha\bigg( v_i(M) - \sum_{j \neq i }v_i(g_{j}) \bigg)\,,\] 
	and thus,
	\[v_i(A_i) \geq \frac{\alpha\big(v_i(M)- \sum_{j \neq i}v_i(g_{j})\big)}{n-1+\alpha} \geq \frac{\alpha\bmu_i}{n-1+\alpha}\,.\]
	To see why the last inequality holds, notice that in any $n$-maximin share defining partition for agent $i$, the bundle that does not contain any of the $g_{j}$s,  has value at most  $v_i(M)- \sum_{j \neq i}v_i(g_{j})$; therefore, $\bmu_i$ is also upper bounded by this quantity.
	
	Regarding tightness, consider the following example. We have an instance with $n \geq 2$ identical agents and $3n-2$ items. Let $\alpha \in (0,1)$ and $V\gg 1$. For every agent $i$ we have 
	\[ 
	v_i(g_j)= \left\{
	\begin{array}{lll}
	\alpha & j=1 \text{ or } j=3k-1, &1 \leq k\leq n-1\\
	1-\alpha &j=3k, &1 \leq k\leq n-1 \\
	V &j=3k+1, &1 \leq k\leq n-1  \\
	\end{array} 
	\right. 
	\]
	Now consider $\mathcal{A} = (A_1, \ldots, A_n)=(\{g_1\}, \{g_2, g_3, g_4\}, \{g_5,g_6, g_7\},\ldots,\{g_{3n-4},g_{3n-3}, g_{3n-2}\})$. It is easy to see that this is an $\alpha$-\efo allocation which is, however, only an $\frac{\alpha}{n-1+\alpha}$-\mms allocation.
\end{proof}

Finally, we investigate the implications that can be derived for \efx and \efo allocations, in terms of \pmms guarantees. Despite Proposition \ref{prop:efx_to_efo} suggesting that (approximate) \efo is much weaker than (approximate) \efx, the two notions give comparable guarantees with respect to \pmms. In particular, the guarantee implied by $\alpha$-\efx can be at most $4/3$ times  the guarantee implied by $\alpha$-\efo.

\begin{proposition}\label{prop:a_efx_to_pmms}
	For $n\ge 2$, any $\alpha$-\efx allocation is also a $\frac{2\alpha}{2+\alpha}$-\pmms allocation for any $\alpha \in (0,1]$, and this is tight.
\end{proposition}

\begin{proof}
Let $\mathcal{A} = (A_1, \ldots, A_n)$ be an $\alpha$-\efx allocation and consider agents $i$ and $j$. We focus on agent $i$. Notice  that when $|A_j|=1$ we have $v_i(A_i) \geq \bmu_i(2,A_i \cup A_j)$ by the definition of maximin share; so assume that $|A_j| \geq 2$. If $g \in \argmin_{h \in A_j}v_i(h)$, then we have $v_i(A_i)\geq \alpha(v_i(A_j)- v_i(g))$ but also $v_i(g)\le 0.5 v_i(A_j)$. Combining these two inequalities we get $v_i(A_i)\geq \frac{\alpha}{2}v_i(A_j)$.
Thus 
\[2\bmu_i(2,A_i \cup A_j) \leq v_i(A_i \cup A_j) =v_i(A_i)+v_i(A_j) \leq  \frac{2+\alpha}{\alpha}v_i(A_i)\,,\]
which directly implies 
\[v_i(A_i) \geq \frac{2\alpha}{2+\alpha}\bmu_i(2,A_i \cup A_j)\,.\]

Regarding tightness, consider the following example: We have an instance with $n \geq 2$ identical agents and $2n$ items. Let $\alpha \in (0,1]$. For every agent $i$ we have 
	\[ 
	v_i(g_j)= \left\{
	\begin{array}{ll}
           \frac{\alpha}{2} & j=1,2\\
	1 & 3 \leq j\leq 2n\\
	\end{array} 
	\right. 
	\]
Now consider $\mathcal{A} = (A_1, \ldots, A_n)=(\{g_1,g_2\}, \{g_3, g_4,\}, \{g_5,g_6\},\ldots,\{g_{2n-1},g_{2n}\})$. It is easy to see that this is an $\alpha$-\efx allocation which is also an $\frac{2\alpha}{2+\alpha}$-\pmms allocation.
\end{proof}

\begin{proposition}\label{prop:efo_to_pmms}
	For $n\ge 3$, 
	any $\alpha$-\efo allocation is also an $\frac{\alpha}{1+\alpha}$-\pmms allocation  for any $\alpha \in (0,1]$, and this is tight.
\end{proposition}

\begin{proof}
Let $\mathcal{A} = (A_1, \ldots, A_n)$ be an $\alpha$-\efo allocation and consider agents $i, j$. We focus on agent $i$ and we assume that $g$ is an item in $A_j$ so that $v_i(A_i) \geq \alpha(v_i(A_j)-v_i(g))$. We split the proof into two cases: \smallskip

\noindent\textit{Case 1:} If $v_i(g) > v_i((A_i \cup A_j)\mysetminus\{g\})$, then $(\{g\}, (A_i \cup A_j)\mysetminus\{g\})$ is a 2-maximin share defining partition for agent $i$. That is, $\bmu_i(2,A_i \cup A_j)=v_i(A_i)+v_i(A_j)-v_i(g)$. Therefore, $\bmu_i(2,A_i \cup A_j)\leq v_i(A_i)+\frac{1}{\alpha}v_i(A_i)$, or equivalently, $v_i(A_i) \geq \frac{\alpha}{1+\alpha}\bmu_i(2,A_i \cup A_j)$.\medskip

\noindent\textit{Case 2:} If $v_i(g) \leq v_i((A_i \cup A_j)\mysetminus\{g\})$, then 
	\begin{IEEEeqnarray*}{rCl}
	2\bmu_i(2,A_i \cup A_j) & \le & v_i(A_i\cup A_j) =  v_i(A_i)+v_i(A_j\mysetminus\{g\})+v_i(g) \\
	& \le & 2v_i(A_i)+2v_i(A_j\mysetminus\{g\}) \leq \left( 2+\frac{2}{\alpha}\right) v_i(A_i) \,,
\end{IEEEeqnarray*} 
or equivalently, $v_i(A_i) \geq \frac{\alpha}{1+\alpha}\bmu_i(2,A_i \cup A_j)$.\smallskip

Regarding tightness, consider the following example: We have an instance with $n \geq 3$ agents and $n+1$ items. Let $\alpha \in (0,1)$ and $V\gg1+\alpha$. We focus on agent $i$ and we have 
	\[ 
	v_1(g_j)= \left\{
	\begin{array}{ll}
           \alpha & j=1 \text{ or } 4 \leq j\leq n+1\\
	1 & j=2\\
	V&j=3 
	\end{array} 
	\right. 
	\]
Now consider $\mathcal{A} = (A_1, \ldots, A_n)=(\{g_1\}, \{g_2, g_3\}, \{g_5\},\ldots,\{g_{n+4}\})$ and assume that agents $2\leq i\leq n$ are not envious. It is easy to see that accoriding to agent 1, this is an $\alpha$-\efo allocation which is also  an $\frac{\alpha}{1+\alpha}$-\pmms allocation.
\end{proof}

\paragraph{Algorithmic Implications for PMMS.} The last two propositions also have further consequences. First of all, it is known that for additive valuations, \efo allocations can be computed efficiently by a simple round-robin algorithm \cite{LMMS04,CKMP016}. Hence Proposition \ref{prop:efo_to_pmms} yields a 1/2-\pmms allocation in polynomial time. Moreover, exact \efx allocations can be computed efficiently when agents have the same ordering on the values of the goods \cite{BM17,PR18}, i.e., when all agents have the same ordinal preferences (but possibly different cardinal values). This implies a $2/3$-approximation by Proposition \ref{prop:a_efx_to_pmms}. These facts are summarized below.

\begin{corollary}\label{col:rr}
	For $n\ge 3$,
	\begin{enumerate}[label=(\arabic*)]
		\item  the round-robin algorithm, where each agent picks in her turn  her favorite available item, produces a $\frac{1}{2}$-\pmms allocation.
		\item  we can compute in polynomial time a $\frac{2}{3}$-\pmms allocation when all agents agree on the ordering of the goods with respect to their value.
	\end{enumerate}	
\end{corollary}

It is an interesting open problem to compute a better than $1/2$-\pmms allocation for additive valuations. So far, the existence of $0.618$-\pmms allocations has been established but without an efficient algorithm \cite{CKMP016}.

\section{PMMS and MMS Allocations}
\label{sec:(P)MMS}

We now explore analogous questions with Section \ref{sec:EFX-EF1}, but starting now with (approximate) \mms or \pmms allocations. 
We begin with the guarantees that \pmms allocations imply for \mms and vice versa. In order to proceed, we will make use of the following observation, that \pmms implies \efx.

\begin{proposition}[Caragiannis et al.~\cite{CKMP016}]\label{prop:pmms_to_efx}
	For $n\ge 2$, any \pmms allocation is also an \efx allocation.
\end{proposition}

For two agents, it is clear by their definition that the notions of \mms and \pmms are identical. For more agents, by Propositions \ref{prop:efx_to_mms_n=2,3}, \ref{prop:efx_to_mms} and \ref{prop:pmms_to_efx}, we have the following corollary.

\begin{corollary} \label{cor:pmms_to_mms}
	For $n = 3$, a \pmms allocation is also a $\frac{2}{3}$-\mms allocation. Moreover, for $n\geq 4$, a \pmms allocation is also a $\frac{4}{7}$-\mms allocation.
\end{corollary}

Moreover, the guarantees of Corollary \ref{cor:pmms_to_mms} are tight for a small number of agents and almost tight for bigger instances.

\begin{proposition}\label{prop:pmms_to_mms}
	For $n\ge 3$, a \pmms allocation is not necessarily an $\alpha$-\mms allocation for $\alpha > \frac{2}{3}$ and, as $n$ grows large, for $\alpha \ge 0.5914$.
\end{proposition}

\begin{proof}
	Suppose that we have  the following instance with 3 agents and 6 items. The items have the following values for agent 1: 
	\[ 
	v_1(g_i)= \left\{
	\begin{array}{ll}
	1/2 & 1\le i \le 3 \\
	1 & 4 \leq i\leq 6 
	\end{array} 
	\right. 
	\]
	Clearly, $\bmu_1 = \frac{3}{2}$. Consider the allocation $\mathcal{A} = (\{g_4\}, \allowbreak \{g_1, g_2, g_3\}, \allowbreak \{g_5, g_6\})$. Assuming that agents $2$ and $3$ have large value for their bundles and negligible value for everything else, it is easy to check that $\mathcal{A}$ is a \pmms allocation. However, $v_1(A_1) = 1 = \frac{2}{3} \bmu_1$. By adding an arbitrary number of copies of agent $3$ and her bundle, this instance can be generalized to any number of agents.
	
	Next, we show a stronger bound as $n$ grows large. The construction below achieves the desired bound for $n \ge 3\cdot 7\cdot 43 \cdot 1806 = 1,631,721$. However, there is a smooth transition from $2/3$ to that; e.g., already for $n\ge 21$ we get a bound of $3/5$, which worsens as $n$ increases.
	We are going to consider the suggested allocation from the viewpoint of 
	agent 1, while assuming that agents $2$ through $n$ have large value for their bundles and negligible value for everything else.
	Since there is a large number of items, we are not going to define $v_1(\cdot)$ explicitly, but implicitly through the different types of bundles
	the agents get. So consider the allocation $\mathcal{A}$ where:
	\begin{itemize}[label={--},itemsep=0.5ex]
		\item  agent $1$ receives 1 item of value $1$,
		\item  $\left\lfloor \frac{n}{3} \right\rfloor$ agents receive 3 items of value $\frac{1}{2}$ each,
		\item  $\left\lfloor \frac{n}{7} \right\rfloor$ agents receive 7 items of value $\frac{1}{6}$ each,
		\item  $\left\lfloor \frac{n}{43} \right\rfloor$ agents receive 43 items of value $\frac{1}{42}$ each,
		\item  $\left\lfloor \frac{n}{1807} \right\rfloor$ agents receive 1807 items of value $\frac{1}{1806}$ each,
		\item  the remaining (at least $\frac{n-1}{2}$) agents receive 2 items of value $1$ each.
	\end{itemize}
	It is easy to see that $\mathcal{A}$ is \efx. What may not be obvious is that the number of agents receiving 2 items is at least $\frac{n-1}{2}$. However, it is a matter of simple calculations to check that $\left( \left\lfloor \frac{n}{3} \right\rfloor + \left\lfloor \frac{n}{7} \right\rfloor + \left\lfloor \frac{n}{43} \right\rfloor + \left\lfloor \frac{n}{1807} \right\rfloor \right) +1 \le \frac{n+1}{2}$ for $n \ge 3$. We show now how to get the bound for  $n = 1,631,721$ (in which case the agents receiving 2 items are exactly $\frac{n-1}{2}$). Then, this instance can be generalized to any number of agents by just adding more agents who receive 2 items of value $1$ each. 
	
	Calculating now $\bmu_1$ is straightforward: in  an $n$-maximin share defining partition for agent $1$, each agent would receive exactly one item of each type. That is, $\bmu_1 = 1 + \frac{1}{2} + \frac{1}{6} + \frac{1}{42} + \frac{1}{1806}$, and $v_1(A_1) = 1 < 0.5914\cdot \bmu_1$. 
\end{proof}

We continue with the worst-case guarantee we can get for \mms by an $\alpha$-\pmms allocation, with $\alpha< 1$. Notice that now the guarantee degrades with $n$.

\begin{proposition}\label{prop:a_pmms_to_mms}
	For $n\geq 3$ and $\alpha \in (0,1)$, any $\alpha$-\pmms allocation is also an $\frac{\alpha}{2(n-1)-\alpha(n-2)}$-\mms allocation but not necessarily a $\beta$-\mms allocation, for any $\beta> \frac{\alpha}{n-1-\alpha(n-2)}$. 
\end{proposition}

\begin{proof}
	The positive result is a direct corollary of Propositions \ref{prop:a_pmms_to_efo} and \ref{prop:efo_to_mms}.
	
	For the upper bound, suppose that we have $n \ge 3$ agents and $2n -1$ items. We focus on agent $1$ and her values 
	\[ 
	v_1(g_j)= \left\{
	\begin{array}{lll}
	\alpha & j=1 & \\
	1-\alpha & j=2k, & 1 \leq k\leq n-1\\
	V & j=2k+1, & 1 \leq k\leq n-1 \\
	\end{array} 
	\right. 
	\]
	where $\alpha \in (0,1)$ and $V\gg 1$. 
	Now consider $\mathcal{A} = (A_1, \ldots, A_4)=(\{g_1\}, \{g_2,g_3\}, \allowbreak \{g_4,g_5\}, \allowbreak \ldots, \allowbreak \{g_{2n-2}, g_{2n-1},\})$. That is, every agent $i>1$ gets one item of value $1-\alpha$ and one item of value $V$.  It is easy to see that this is an $\alpha$-\pmms allocation, given that all agents, other than agent $1$, have large value for their bundles and negligible value for everything else. However, $\bmu_{1}(n, M) = \alpha + (n-1)(1-\alpha) = n-1-\alpha(n-2)$, and therefore $\mathcal{A}$ ia only an $\frac{\alpha}{n-1-\alpha(n-2)}$-\mms allocation.
\end{proof}

The next result exhibits a sharp contrast between \pmms and \mms. Although \pmms allocations (exact or approximate) imply some  \mms guarantee, even exact \mms allocations do not imply any approximation with respect to \pmms.

\begin{proposition}\label{prop:a_mms_to_pmms}
	For $n\ge 3$, an $\alpha$-\mms allocation is not necessarily a $\beta$-\pmms allocation for any  $\alpha, \beta \in (0,1]$. 
\end{proposition}

\begin{proof}
	It suffices to prove the proposition for $\alpha = 1$. Consider an instance with $n \geq 3$  agents and $2n-2$ items. We focus on agent $1$ and her values 
	\[ 
	v_1(g_j)= \left\{
	\begin{array}{ll}
	\epsilon & j=1\\
	1 & 2 \leq j\leq n\\
	0 & n+1 \leq j\leq 2n-2 
	\end{array} 
	\right. 
	\]
	where $\epsilon>0$ is arbitrarily small. That is, agent $1$ views $n-1$ items as ``large'', $1$ item as ``very small'', and everything else as ``worthless''. As a consequence $\bmu_{1}(n, M)=\epsilon$. Now consider the allocation $\mathcal{A}=(A_1, \ldots, A_n)=(\{g_1\}, \allowbreak \{g_{2},\ldots,g_{n}\}, \allowbreak \{g_{n+1}\}, \ldots, \{g_{2n-2}\})$ and 
	assume any agent $i\neq 1$ has large value for her bundle and negligible value for everything else. 
	It is easy to see that $\mathcal{A}$  is an \mms allocation. However $\bmu_{1}(2, A_1\cup A_2)\ge \left\lfloor\frac{n-1}{2}\right\rfloor \ge 1$ and thus  this is only a 
	$\gamma$-\pmms allocation for $\gamma\le\epsilon$.
\end{proof}

We conclude this section by discussing the guarantees of \mms and \pmms with respect to  \efo and \efx. Even for \efo, we see that \mms and \pmms differ significantly in what they can achieve.

\begin{proposition}\label{prop:a_pmms_to_efo}
	For $n\ge 2$, an $\alpha$-\pmms  allocation is also an $\frac{\alpha}{2-\alpha}$-\efo allocation for any  $\alpha \in (0,1)$, and this is tight.
\end{proposition}

\begin{proof}
	Let $\mathcal{A} = (A_1, \ldots, A_n)$ be an $\alpha$-\pmms allocation. Consider agents $i, j$ and let $v_i(A_i)= \gamma \bmu_i(2, A_i \cup A_j)$ for some $\gamma \ge \alpha$. We focus on agent $i$, so, without loss of generality, we may assume that there is no item $h$ such that $v_i(h)=0$. Notice that when $|A_j|\le 1$, either there is nothing to check or we need to check that agent $i$ does not envy the empty bundle (which is, of course, true). So, we assume that $|A_j| \geq 2$. 
	Also, we may assume that $\gamma\le 1$, or else $v_i(A_j)< \bmu_i(2, A_i \cup A_j)$ and agent $i$ does not envy agent $j$ anyway.  
	Suppose that $g \in \argmax_{h \in A_j}v_i(h)$ and consider the following two cases: \medskip
	
	\noindent\textit{Case 1}: $v_i(g) \geq \bmu_i(2, A_i \cup A_j)$. Notice that by the definition of $\bmu_i(2, A_i \cup A_j)$, we must have that $\bmu_i(2, A_i \cup A_j)= v_i(A_i)+v_i(A_j)-v_i(g)$. Thus,  
	$v_i(A_j)-v_i(g)= \bmu_i(2, A_i \cup A_j) - v_i(A_i) = (1-\gamma)\bmu_i(2, A_i \cup A_j)$.
	If $\gamma =1$ then clearly $v_i(A_i)\ge v_i(A_j)-v_i(g) = 0$. 
	Otherwise, we have  
	\[v_i(A_i)= \gamma \bmu_i(2, A_i \cup A_j) =\frac{\gamma}{1-\gamma} (v_i(A_j)-v_i(g)) \ge \frac{\alpha}{1-\alpha} (v_i(A_j)-v_i(g))\,.\]
	
	\noindent\textit{Case 2}: $v_i(g)=\beta < \bmu_i(2, A_i \cup A_j)$. Initially we prove that $v_i(A_j)\le\beta + 2\bmu_i(2, A_i \cup A_j)-\gamma\bmu_i(2, A_i \cup A_j)$. Suppose for a contradiction that $v_i(A_j) > \beta + 2\bmu_i(2, A_i \cup A_j)-\gamma\bmu_i(2, A_i \cup A_j)$.

	
	Now we sort the items in $A_j$ in an increasing order according to their value regarding agent $i$, say $h_1, h_2,\ldots, h_k=g$. Start with set $A_i$ and add items from set $A_j$, one at a time, until the value exceeds $\bmu_i(2, A_i \cup A_j)$, i.e., find an index $l$ such that $v_i(A_i \cup \{h_1, h_2,\ldots,h_{l}\})>\bmu_i(2, A_i \cup A_j)$ but $v_i(A_i \cup \{h_1, h_2,\ldots,h_{l-1}\})\leq\bmu_i(2, A_i \cup A_j)$ (slightly abusing notation here, $\{h_1, \ldots,h_{l-1}\}$ denotes the empty set when $l=1$). Let  $S=A_i\cup \{h_1, h_2,\ldots,h_{l}\}$. Since $v_i(S\mysetminus \{h_l\})\le\bmu_i(2, A_i \cup A_j)$ and $v_i(h_l) \le v_i(g)=\beta$ we have that $v_i(S)\le\bmu_i(2, A_i \cup A_j)+\beta$ and thus $v_i(\{h_1, h_2,\ldots,h_{l}\}) \le \bmu_i(2, A_i \cup A_j)+\beta-v_i(A_i)$. Therefore
	\begin{IEEEeqnarray*}{rCl}
		v_i(A_j\mysetminus S) & = & v_i(A_j)-  v_i(\{h_1, h_2,\ldots,h_{l}\}) \ge  v_i(A_j)-\bmu_i(2, A_i \cup A_j)-\beta+v_i(A_i)\\
		& > & 2\bmu_i(2, A_i \cup A_j)+\beta-\gamma\bmu_i(2, A_i \cup A_j)-\bmu_i(2, A_i \cup A_j)-\beta+ v_i(A_i)\\
		& = & \bmu_i(2, A_i \cup A_j) \,.
	\end{IEEEeqnarray*} 
	%
	%
	This means that in the partition $(S, A_j\mysetminus S)$  both sets are strictly better than $\bmu_i(2, A_i \cup A_j)$ for agent $i$; this contradicts the definition of $\bmu_i(2, A_i \cup A_j)$. 
	
	So it must be the case where $v_i(A_j)\le\beta + 2\bmu_i(2, A_i \cup A_j)-\gamma\bmu_i(2, A_i \cup A_j)$. Therefore $v_i(A_j)-v_i(g)=v_i(A_j)-\beta \le (2-\gamma)\bmu_i(2, A_i \cup A_j)$, and so 
	\[v_i(A_i)=\gamma\bmu_i(2, A_i \cup A_j)\ge\frac{\gamma}{2-\gamma} (v_i(A_j)-v_i(g))\ge \frac{\alpha}{2-\alpha} (v_i(A_j)-v_i(g))\,.\]
	
	We conclude that an $\alpha$-\pmms allocation leads to an $\frac{\alpha}{2-\alpha}$-\efo allocation since $\frac{\alpha}{2-\alpha} \le \frac{\alpha}{1-\alpha}$ for any $\alpha \in(0,1)$.
	
	Regarding tightness, consider the following example: We have an instance with $n=2$ agents and $2k+1\geq 5$ items. Let $\alpha \in (0,1)$. We focus on agent $1$ and we have 
	\[ 
	v_1(g_j)= \left\{
	\begin{array}{ll}
	\frac{\alpha}{k} & 1 \leq j\leq k \vspace{2pt}\\
	\frac{2-\alpha}{k} &k+1 \leq j\leq 2k  \\
	\epsilon &j=2k+1  
	\end{array} 
	\right. 
	\]
	where $\epsilon$ is an arbitrarily small positive amount. Now consider $\mathcal{A} = (A_1, \ldots, A_n)=(\{g_1,\ldots,g_k\}, \{g_{k+1}, \allowbreak \ldots, \allowbreak g_{2k+1}\})$ and assume that agents $2$ is not envious. It is easy to see that accoriding to agent 1, this is an $\alpha$-\pmms allocation which is also however, only an $\frac{\alpha}{2-\alpha}$-\efo allocation. It is straightforward to generalize this instance to any number of agents.
\end{proof}

\begin{proposition}\label{prop:a_mms_to_efo}
	For $n\ge 3$, an $\alpha$-\mms  allocation is not necessarily a $\beta$-\efo allocation for any  $\alpha, \beta \in (0,1]$. 
\end{proposition}

\begin{proof}
We may use the  same example as in Proposition \ref{prop:a_mms_to_pmms}. Like before, $\mathcal{A}$ is an \mms allocation but  only a $\gamma$-\efo allocation for $\gamma\le\epsilon$. To see this, notice that $v_1(A_2\mysetminus\{g\})=n-2\ge 1$ for any $g\in A_2$.
\end{proof}

When one focuses on guarantees in terms of \efx, there are  only negative results. In fact, it is rather surprising that even though \pmms implies \efx, an $\alpha$-\pmms allocation with $\alpha<1$ does not imply \textit{any} approximation in terms of \efx.

\begin{proposition}\label{prop:a_pmms_to_efx}
	For $n\ge 2$, an $\alpha$-\pmms  allocation is not necessarily a $\beta$-\efx allocation for any  $\alpha, \beta \in (0,1)$.
\end{proposition}

\begin{proof}
	Consider an instance with $n \geq 2$ identical agents and $n+1$ items. Let $\alpha \in (0,1)$ and $V\gg max\{1, \frac{1}{\alpha}-1\}$. For every agent $i$ we have 
	\[ 
	v_i(g_j)= \left\{
	\begin{array}{ll}
	V & j=1\\
	\frac{1}{\alpha}-1 & j=2\\
	1 & 3 \leq i\leq n+1 \\
	\end{array} 
	\right. 
	\]
	We examine the allocation $\mathcal{A}=(A_1, \ldots, A_n)=(\{g_1, g_2\}, \allowbreak \{g_{3}\}, \allowbreak \{g_{4}\}, \ldots, \{g_{n+1}\})$. For any $i\neq 1$, we focus on agents $i$ and $1$ from $i$'s viewpoint. It is easy to see that $\bmu_{i}(2, A_1 \cup A_{i})=\frac{1}{\alpha}-1+1=\frac{1}{\alpha}$.  Since $v_i(A_i)=1$, we get  $v_i(A_i)=\alpha \bmu_{i}(2, A_1 \cup A_{i})$. 
	As it is straightforward  to see that for any other pair of agents that there is no envy, we have that $\mathcal{A}$ is an $\alpha$-\pmms allocation. On the other hand, every agent $i \in N\mysetminus\{1\}$ envies agent $1$ up to any item by a factor $\gamma=\frac{1}{V}$ which can become arbitrarily small for sufficiently large $V$.
\end{proof}

By Propositions \ref{prop:a_pmms_to_efx},  \ref{prop:a_mms_to_efo}, and \ref{prop:efx_to_efo}, we also have the corresponding result for approximate \mms allocations. 
\begin{corollary}\label{cor:a_mms_to_efx}
	For $n\ge 2$, an $\alpha$-\mms allocation is not necessarily a $\beta$-\efx allocation for any  $\alpha, \beta \in (0,1)$.
\end{corollary}

\section{Discussion}
\label{sec:concl}

The main purpose of this work is to provide a deeper understanding of the connections between the exact and the approximate versions of four prominent fairness notions used in discrete models of fair division. In most cases we give tight, or almost tight, results, providing therefore a close to complete picture on the status of these questions. 
Some of our findings also strike as counter-intuitive, given what was known for the exact versions of these notions. 

A direct implication of our results is the non-existence of \emph{truthful} allocation mechanisms for two agents that achieve any constant approximation of \pmms \efx, or even \efo. This follows from the corresponding negative result of Amanatidis et al.~\cite{ABCM17} for \mms and Propositions \ref{prop:a_efx_to_mms}, \ref{prop:efo_to_mms}, and \ref{prop:a_pmms_to_mms}. We defer a more detailed discussion on this to the full version of the paper.

Aside from the questions raised here, there is an abundance of interesting open problems for future research. Deciding the existence of exact \efx or exact \pmms allocations seem to be two of the most puzzling such problems. Furthermore, obtaining (algorithmically or existentially) allocations with better approximation ratios is another class of equally intriguing problems. So far, we know there exist allocations with ratios of $0.75$, $0.618$, and $0.5$, for \mms, \pmms, and \efx respectively, out of which, the result for \pmms is existential; our Corollary \ref{col:rr} only yields a $0.5$-approximation in polynomial-time.
Finally,  the combination of fairness with other desired objectives, like Pareto optimality, creates further algorithmic challenges, even for the seemingly easier notion of \efo \cite{BMV18}.  

Interestingly enough, our results suggest that improving the current state of the art in the approximation of one of these notions does not imply an immediate improvement on the best approximation ratio for the others, with the notable exception of $\alpha$-\efx. An algorithmic result with $\alpha>2/3$ or an existential result with $\alpha>0.895$  for $\alpha$-\efx would imply an improved algorithmic or existential result for $\beta$-\pmms.

\newpage

\bibliographystyle{plain}
\bibliography{fairdivrefs_ijcai_full}

\begin{thebibliography}{10}

\bibitem{ABCM17}
Georgios Amanatidis, Georgios Birmpas, George Christodoulou, and Evangelos
  Markakis.
\newblock Truthful allocation mechanisms without payments: Characterization and
  implications on fairness.
\newblock In {\em Proceedings of the 18th {ACM} Conference on Economics and
  Computation (EC)}, pages 545--562, 2017.

\bibitem{AMNS17}
Georgios Amanatidis, Evangelos Markakis, Afshin Nikzad, and Amin Saberi.
\newblock Approximation algorithms for computing maximin share allocations.
\newblock {\em {ACM} Trans. Algorithms}, 13(4):52:1--52:28, 2017.
\newblock \emph{A preliminary conference version appeared in ICALP 2015}.

\bibitem{BBMN17}
Siddharth Barman, Arpita Biswas, Sanath Kumar~Krishna Murthy, and Yadati
  Narahari.
\newblock Groupwise maximin fair allocation of indivisible goods.
\newblock {\em 32nd AAAI Conference on Artificial Intelligence, {AAAI} 2018},
  2018.

\bibitem{BM17}
Siddharth Barman and Sanath Kumar~Krishna Murthy.
\newblock Approximation algorithms for maximin fair division.
\newblock In {\em Proceedings of the 18th {ACM} Conference on Economics and
  Computation (EC)}, pages 647--664, 2017.

\bibitem{BMV18}
Siddharth Barman, Sanath Kumar~Krishna Murthy, and Rohit Vaish.
\newblock Finding fair and efficient allocations.
\newblock In {\em 19th ACM Conference on Economics and Computation (EC 2018)},
  2018.
\newblock To appear.

\bibitem{BCM16-survey}
Sylvain Bouveret, Yann Chevaleyre, and Nicolas Maudet.
\newblock Fair allocation of indivisible goods.
\newblock In F.~Brandt, V.~Conitzer, U.~Endriss, J.~Lang, and A.D. Procaccia,
  editors, {\em Handbook of Computational Social Choice}, chapter~12. Cambridge
  University Press, 2016.

\bibitem{BL16}
Sylvain Bouveret and Michel Lema{\^{\i}}tre.
\newblock Characterizing conflicts in fair division of indivisible goods using
  a scale of criteria.
\newblock {\em Autonomous Agents and Multi-Agent Systems}, 30(2):259--290,
  2016.
\newblock \emph{A preliminary conference version appeared in AAMAS 2014}.

\bibitem{BT96}
Steven~J. Brams and Alan~D. Taylor.
\newblock {\em Fair Division: from Cake Cutting to Dispute Resolution}.
\newblock Cambridge University press, 1996.

\bibitem{Budish11}
Eric Budish.
\newblock The combinatorial assignment problem: Approximate competitive
  equilibrium from equal incomes.
\newblock {\em Journal of Political Economy}, 119(6):1061--1103, 2011.

\bibitem{CKMP016}
Ioannis Caragiannis, David Kurokawa, Herv{\'{e}} Moulin, Ariel~D. Procaccia,
  Nisarg Shah, and Junxing Wang.
\newblock The unreasonable fairness of maximum {N}ash welfare.
\newblock In {\em Proceedings of the 17th {ACM} Conference on Economics and
  Computation (EC)}, pages 305--322, 2016.

\bibitem{Foley67}
Duncan~K. Foley.
\newblock Resource allocation and the public sector.
\newblock {\em Yale Economics Essays}, 7:45--98, 1967.

\bibitem{GS58}
George Gamov and Marvin Stern.
\newblock {\em Puzzle-Math}.
\newblock Viking press, 1958.

\bibitem{SGHSY18}
Mohammad Ghodsi, Mohammad~Taghi Hajiaghayi, Masoud Seddighin, Saeed Seddighin,
  and Hadi Yami.
\newblock Fair allocation of indivisible goods: Improvement and generalization.
\newblock In {\em 19th ACM Conference on Economics and Computation (EC 2018)},
  2018.
\newblock To appear.

\bibitem{KPW16}
David Kurokawa, Ariel~D. Procaccia, and Junxing Wang.
\newblock When can the maximin share guarantee be guaranteed?
\newblock In {\em 30th AAAI Conference on Artificial Intelligence, {AAAI}
  2016}, pages 523--529, 2016.

\bibitem{KurokawaPW18}
David Kurokawa, Ariel~D. Procaccia, and Junxing Wang.
\newblock Fair enough: Guaranteeing approximate maximin shares.
\newblock {\em J. {ACM}}, 65(2):8:1--8:27, 2018.

\bibitem{LMMS04}
Richard~J. Lipton, Evangelos Markakis, Elchanan Mossel, and Amin Saberi.
\newblock On approximately fair allocations of indivisible goods.
\newblock In {\em ACM Conference on Electronic Commerce (EC)}, pages 125--131,
  2004.

\bibitem{Markakis17-survey}
Evangelos Markakis.
\newblock Approximation algorithms and hardness results for fair division with
  indivisible goods.
\newblock In U.~Endriss, editor, {\em Trends in Computational Social Choice},
  chapter~12. AI Access, 2017.

\bibitem{Moulin90}
Herv{\'{e}} Moulin.
\newblock Uniform externalities: Two axioms for fair allocation.
\newblock {\em Journal of Public Economics}, 43(3):305--326, 1990.

\bibitem{Moulin03}
Herv{\'{e}} Moulin.
\newblock {\em Fair division and collective welfare}.
\newblock {MIT} Press, 2003.

\bibitem{PR18}
Benjamin Plaut and Tim Roughgarden.
\newblock Almost envy-freeness with general valuations.
\newblock In {\em 29th Annual {ACM-SIAM} Symposium on Discrete Algorithms,
  {SODA} 2018}, pages 2584--2603. {SIAM}, 2018.

\bibitem{Procaccia16-survey}
Ariel~D. Procaccia.
\newblock Cake cutting algorithms.
\newblock In F.~Brandt, V.~Conitzer, U.~Endriss, J.~Lang, and A.D. Procaccia,
  editors, {\em Handbook of Computational Social Choice}, chapter~13. Cambridge
  University Press, 2016.

\bibitem{RW98}
Jack~M. Robertson and William~A. Webb.
\newblock {\em Cake Cutting Algorithms: be fair if you can}.
\newblock AK Peters, 1998.

\bibitem{Suk18}
Warut Suksompong.
\newblock Approximate maximin shares for groups of agents.
\newblock {\em Mathematical Social Sciences}, 2018.
\newblock To appear.

\bibitem{Varian74}
Hal~R. Varian.
\newblock Equity, envy and efficiency.
\newblock {\em Journal of Economic Theory}, 9:63--91, 1974.

\end{thebibliography}

\end{document}